\documentclass[12pt]{IEEEtran}

\usepackage{amsmath, graphicx, amssymb, amsfonts, subfigure, setspace}

\newtheorem{dfn}{Definition}

\newtheorem{theorem}{Theorem}

\newtheorem{prop}{Proposition}
\newtheorem{lemma}{Lemma}

\begin{document}
\onecolumn
\doublespacing

\title{On the Delay of Network Coding\\ over Line Networks}

\author{
\authorblockN{Theodoros K. Dikaliotis, Alexandros G. Dimakis, Tracey Ho, Michelle Effros}
\authorblockA{Department of Electrical Engineering\\
California Institute of Technology\\
{\tt email: \{tdikal,adim,tho,effros\}@caltech.edu}}}

\newcommand{\stexp}{\mbox{$\mathbb{E}$}}
\newcommand{\Prob}{\ensuremath{\mathbb{P}}}
\newcommand{\bigO}{\mbox{$\mathsf{O}$}}

\maketitle

\begin{abstract}
We analyze a simple network where a source and a receiver are
connected by a line of erasure channels of different reliabilities.
Recent prior work has shown that random linear network coding can
achieve the min-cut capacity and therefore the asymptotic rate is
determined by the worst link of the line network. In this paper we
investigate the delay for transmitting a batch of packets, which is
a function of all the erasure probabilities and the number of
packets in the batch. We show a monotonicity result on the delay
function and derive simple expressions which characterize the
expected delay behavior of line networks. Further, we use a
martingale bounded differences argument to show that the actual
delay is tightly concentrated around its expectation.
\end{abstract}

\section{Introduction}

A common approach for practical network coding performs random
linear coding over batches or generations~\cite{chou03practical},
where the relevant delay measure is the time taken for the batch to
be received. Such in-network coding is particularly beneficial in
lossy networks~\cite{lun04coding} compared to end-to-end erasure
coding. In this paper we investigate the batch end-to-end
delay for lossy line networks. We consider the use of random linear
network coding without feedback and a packet erasure model with
different link qualities. All the nodes in the network store all the
packets they receive and whenever given a transmission opportunity,
send a random linear combination of all the stored
packets~\cite{lun04coding,pakzad05coding} over erasure links.

Despite the extensive recent work on network coding over lossy
networks (e.g.~\cite{lun04coding,pakzad05coding,dana06capacity}) the
expected time required to send a fixed number of packets over a
network of erasure links is not completely characterized. Closely
related work on delay in queueing
theory~\cite{rubin74communication,shalmon87exact} assumes Poisson
arrivals and their results pertain to the delay of individual packets in steady state and \cite{ephremides2006} examines the delay for a single queue multicasting to several users using block network coding. In our work,
we consider a batch of $n$ packets that need to be
communicated over a  line network of $\ell$ erasure links where each link
experiences an erasure with probability $p_1, p_2,\ldots, p_\ell$
and we are interested in the expected total time $\stexp T_n$ for
the $n$ packets to travel across the line network.

Prior work~\cite{lun04coding,pakzad05coding} established that random linear network coding can achieve the min-cut capacity and therefore the asymptotic rate is determined by the worst link of the line network. Therefore, the expected time $\stexp T_n$ for the $n$ packets to cross the network is
\begin{equation}
\label{delay_fun_1}
\stexp T_n  =  \frac{n}{1-\displaystyle\max_{1\leq i\leq \ell} p_i}+ D(n,p_1,p_2,\ldots, p_\ell),
\end{equation}
where the delay function $D(n,p_1,p_2,\ldots, p_\ell)$ is the sublinear part:
\begin{equation*}
\lim_{n\rightarrow \infty, \ell\, \text{fixed}} \frac{D(n,p_1,p_2, \ldots, p_\ell)}{n}=0.
\end{equation*}
However, relatively little is known about the delay function
$D(n,p_1,p_2,\ldots, p_\ell)$.

In this work we characterize the delay function by showing that it is
non-decreasing in $n$ and is bounded by a simple function $\bar
D(p_1,p_2,\ldots, p_\ell)$ of the link erasure probabilities. The main results of this paper are the following two theorems which characterize the expected behavior and show a concentration of the actual delay random variable close to this expectation.

\begin{theorem}
Consider $n$ packets communicated through a line network of $\ell$ links with erasure probabilities $p_1,p_2,\ldots, p_\ell$ and assume that there is a unique worst link:
\begin{equation*}
p_m := \displaystyle\max_{1\leq i\leq \ell} p_i, \quad p_i <p_m<1 \quad \forall \, i \neq m.
\end{equation*}
The expected time $\stexp T_n$ to send all $n$ packets is:
\begin{equation*}
\label{delay_fun_2}
\stexp T_n  =  \frac{n}{1-\displaystyle\max_{1\leq i\leq \ell} p_i}+ D(n,p_1,p_2,\ldots, p_\ell),
\end{equation*}
where the delay function $D(n,p_1,p_2,\ldots, p_\ell)$ is non-decreasing in $n$ and upper bounded by:
\begin{equation*}
\bar D(p_1,p_2, \ldots, p_\ell) :=
\sum_{i=1, i \neq m}^{\ell}\frac{p_m}{p_m - p_i}.
\end{equation*}
\label{thm:Expected_values_multi_hop_network}
\end{theorem}
If on the other hand there are two links that take the worst value,
then the delay function is not bounded but still exhibits the
sublinear behavior. Pakzad et al.~\cite{pakzad05coding} prove that
in the case of a two-hop network with identical links the delay function
grows as $\sqrt{n}$. We also prove the following concentration result:

\begin{theorem}
The time $T_n$ for $n$ packets to travel across the network is concentrated around its expected value with high probability. In particular for sufficiently large $n$:
\begin{equation*}
\Prob\left[\left|T_n-\stexp T_n\right|> \epsilon_{n}\right] \leq \frac{2\left(1-\displaystyle\max_{1\leq i\leq \ell} p_i\right)}{n}+  o\left(\frac{1}{n}\right),
\label{eqn:theorem2}
\end{equation*}
for deviations $\epsilon_{n} = n^{3/4}/(1-\displaystyle\max_{1\leq i\leq \ell} p_i)$.
\label{thm:theorem2}
\end{theorem}

Since $\stexp T_n$ grows linearly in $n$ and the deviations $\epsilon_n$ are sublinear, $T_n$ is tightly concentrated around its expectation for large $n$ with probability approaching one.

The remainder of this paper is organized as follows: Section~\ref{The_Model} presents the precise model we use for packet communication. Section~\ref{General_Line_network} presents the analysis for the general multi-hop network. Section~\ref{Discussion} contains a discussion of the results presented in this paper along with comments for future research.    

\section{Model}
\label{The_Model}

The general network under consideration is depicted in Fig.~\ref{fig:Multi_hop_network}. The network consists of $\ell+1$ nodes $N^{(i)},1\leq i \leq \ell+1,$ and $\ell$ links $L^{(i)}$,$1\leq i \leq \ell$, with source node $N^{(1)}$ and destination node $N^{(\ell+1)}$. Node $N^{(i)}, 1\leq i \leq \ell$ is connected to node $N^{(i+1)}$ to its right through the erasure link $L^{(i)}$.

We assume a discrete time model in which the source wishes to transmit $n$ packets to the destination. At each time step, node $N^{(i)}$ can transmit one packet through link $L^{(i)}$ to node $N^{(i+1)}$, $1\leq i \leq \ell$. The transmission succeeds with probability $1-p_i$ or the packet gets erased with probability $p_i$. Erasures across different links and time steps are assumed to be independent. At each time step the packet transmitted by node $N^{(i)}$ is a random linear combination of all previously received packets at the node. We want to determine the time $T_n$ taken for the destination node to receive (decode) all the $n$ packets initially present at the source node $N^{(1)}$. We assume that no link fails with probability 1 ($p_i<1,1 \leq i\leq \ell$) or else the problem becomes trivial since there are no packets traveling through the network. The destination node $N^{(\ell+1)}$ will decode once it receives $n$ linearly independent combinations of the initial packets.

Coding at each hop (network coding) is needed to achieve minimum delay when feedback is unavailable, slow or expensive. If instantaneous feedback is available at each hop an automatic repeat request (ARQ) scheme with simple forwarding of packets achieves a block delay performance identical to network coding. Note that coding only at the source is suboptimal in terms of throughput  and delay \cite{lun04coding}. The only feedback required in the network coding case is that the destination node $N^{(\ell+1)}$, once it receives all the necessary linearly independent packets, signals the end of transmission to all the other nodes.

As explained in \cite{jay2007}, information travels through the network in the form of innovative packets. A packet at node $N^{(i)}$, $2\leq i \leq \ell$ is innovative if it does not belong to the space spanned by packets present at node $N^{(i+1)}$. Each node needs to code, and therefore store, only the part of the information that has not already been received by $N^{(i+1)}$. If feedback was present, nodes could equivalently drop packets that do not add information to the nodes on their right. Therefore the analysis becomes essentially a queueing theory problem for innovative packets.

In our model, in case of a success the packet is assumed to be transmitted to the next node instantaneously, i.e. we ignore the transmission delay along the links. Moreover, there is no restriction on the number of packets $n$ or the number of hops $\ell$, and there is no requirement for the network to reach steady state.

\begin{figure}[!ht]
\begin{center}
\includegraphics[clip=true, trim=0mm 0mm 0mm 0mm, width=1.0\columnwidth]{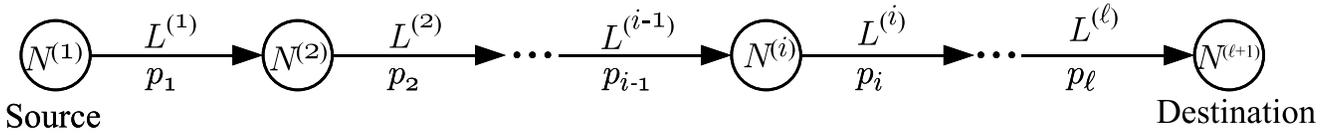}
\end{center}
\caption{Multi-hop network}
\label{fig:Multi_hop_network}
\end{figure}

\section{General Line Networks}
\label{General_Line_network}

\subsection{Proof of Theorem 1}

Let the random variable $R_n^{(i)}, 2\leq i \leq \ell,$ denote the rank difference between node $N^{(i)}$ and node $N^{(i+1)}$, at the  moment packet $n$ arrives at $N^{(2)}$. This is exactly the number of innovative packets present at node $N^{(i)}$ at the random time when packet $n$ arrives at $N^{(2)}$.

The time $T_n$ taken to send $n$ packets from the source node $N^{(1)}$ to the destination $N^{(\ell+1)}$ can be expressed as the sum of time $T_n^{(1)}$ required for all the $n$ packets to cross the first link and the time $\tau_n$ required for all the remaining innovative packets $R_n^{(2)},\ldots,R_n^{(\ell)}$ at nodes $N^{(2)},\ldots,N^{(\ell)}$ respectively to reach the destination node $N^{(\ell+1)}$:
\begin{equation}
T_n = T_n^{(1)}+\tau_n.
\label{eqn:Sum_of_times}
\end{equation}
All the quantities in equation (\ref{eqn:Sum_of_times}) are random variables and we want to compute their expected values. Due to the linearity of the expectation
\begin{equation}
\stexp T_n = \stexp T_n^{(1)} + \stexp \tau_n
\label{eqn:Exp_sum_of_times}
\end{equation}
and by defining $X_j^{(1)}, 1 \leq j \leq n$ to be the time taken for packet $j$ to cross the first link, we get:
\begin{eqnarray}
\stexp T_n^{(1)} =
\sum_{j=1}^{n}\stexp X_j^{(1)} = \frac{n}{1-p_1}
\label{eqn_claim:Expected_values_1}
\end{eqnarray}
since $X_j^{(1)}, 1 \leq j \leq n,$ are all geometric random variables ($\Prob \left(X_j^{(2)}=k\right)=(1-p_1)\cdot p_1^{k-1}, k\geq 1$). Therefore combining equations (\ref{eqn:Exp_sum_of_times}) and (\ref{eqn_claim:Expected_values_1}) we get:
\begin{eqnarray}
\stexp T_n^{(1)} = \frac{n}{1-p_1} + \stexp \tau_n.
\label{eqn_claim:Expected_values_2}
\end{eqnarray}

Equations (\ref{delay_fun_1}), (\ref{eqn_claim:Expected_values_2}) give us
\begin{equation*}
D(n,p_1,p_2,\ldots,p_\ell)=\frac{n}{1-p_1}-\frac{n}{1-\displaystyle\max_{1\leq i\leq \ell} p_i}+\stexp \tau_n
\label{eqn:Delay_function}
\end{equation*}
and clearly the key quantity for calculating the delay function $D(n,p_1,p_2,\ldots,p_\ell)$ is the expected time $\stexp \tau_n$ taken for all the remaining innovative packets at nodes $N^{(2)},\ldots,N^{(\ell)}$ to reach the destination. For the simplest case of a two-hop network ($\ell=2$) we can derive recursive formulas for computing this expectation for each $n$. Table \ref{table:E(Rem_n)} has closed-form expressions for the delay function $D(n, p_1, p_2)$ for $n=1,\ldots,4$.
\begin{table*}[ht]
\centering
\label{table:E(Rem_n)}
\caption{The delay function $D(n, p_1, p_2)$ for different values of $n$}
\begin{tabular}{c c}
\hline\hline
$n$ & $D(n, p_1, p_2)$\\ \hline
1 & $\frac{1}{1-p_1}-\frac{1}{1-\max (p_1, p_2)}+\frac{1}{1-p_2}$ \\ [1.0ex]
2 & $\frac{2}{1-p_1}-\frac{2}{1-\max (p_1, p_2)}+\frac{2}{1-p_2}-\frac{1}{1-p_1 p_2}$  \\  [1.0ex]
3 & $\frac{3}{1-p_1}-\frac{3}{1-\max (p_1, p_2)}+\frac{1+p_2 \left(2-p_1 \left(6-p_1+(2-5 p_1) p_2+(1-3 (1-p_1) p_1) p_2^2\right)\right)}{(1-p_2)(1-p_1 p_2)^3}$  \\ [1.0ex]
4 & $\frac{4}{1-p_1}-\frac{4}{1-\max (p_1, p_2)}+
\frac{\tiny\left\{\begin{array}{c}
1+p_2 (3-p_1 (11+4 p_1^4 p_2^4+p_2 (5+(5-p_2) p_2)+p_1^3 p_2 (1-p_2 (5+2 p_2 (5+3 p_2)))\\
-p_1 (4+p_2 (15+p_2 (21-(1-p_2) p_2)))+p_1^2 (1-p_2 (1-p_2 (31+p_2 (5+4 p_2))))))
\end{array}\right\}}
{(1-p_2)(1-p_1 p_2)^5}$ \\ [1.0ex]
\hline
\end{tabular}
\end{table*}
It is seen that as $n$ grows, the number of terms in the above expression increases rapidly, making these exact formulas impractical, and as expected for larger values of $\ell$ ($\geq 3$) the situation only worsens. Our subsequent analysis derives tight upper bounds on the delay function $D(n,p_1,p_2,\ldots,p_\ell)$ for any $\ell$ which do not depend on $n$.

The $(\ell-1)$-tuple $Y_n = (R_n^{(2)},\ldots,R_n^{(\ell)})$ representing the number of innovative packets remaining at nodes $N^{(2)},\ldots,N^{(\ell)}$ the moment packet $n$ arrives at node $N^{(2)}$ (including packet $n$) is a multidimensional Markov process with state space $E \subset \mathbb{N}\hspace{0.8mm}^{\ell-1}$ (the state space is a proper subset of $\mathbb{N}\hspace{0.8mm}^{\ell-1}$ since $Y_n$ can never take the values $(0,*,\ldots,*)$). Using the coupling method \cite{stock2} and an argument similar to the one given at Proposition 2 in \cite{aggregate} it can be shown that $Y_n$ is a stochastically increasing function of $n$ (meaning that as $n$ increases there is a higher probability of having more innovative packets at nodes $N^{(2)},\ldots,N^{(\ell)}$).

\begin{prop}
The Markov process $Y_n = (R_n^{(2)},\ldots,R_n^{(\ell)})$ is $\preceq_{\text{st}}$-increasing.
\label{prop:The_markov_chain}
\end{prop}
\begin{proof}
Given in the appendix along with the necessary definitions.
\end{proof}

A direct result of Proposition~\ref{prop:The_markov_chain} is that the expected time taken $\stexp \tau_n$ for the remaining packets at nodes $N^{(2)},\ldots,N^{(\ell)}$ to reach the destination is a non-decreasing function of  $n$:
\begin{equation}
\stexp \tau_n \leq \stexp \tau_{n+1} \leq \displaystyle\lim_{n\rightarrow\infty} \stexp \tau_n
\label{eqn:inequality_of_expectations}
\end{equation}
where in the second inequality is meaningful when the limit exists.
  
Innovative packets travelling in the network from node $N^{(2)}$ to node $N^{(\ell+1)}$ can be viewed as customers travelling through a network of service stations in tandem. Indeed, each innovative packet (customer) arrives at the first station (node $N^{(2)}$) with a geometric arrival process and the transmission (service) time is also geometrically distributed. Once an innovative packet has been transmitted (serviced) it leaves the current node (station) and arrives at the next node (station) waiting for its next transmission (service).

By using the interchangeability result on service station from Weber~\cite{weber92interchangeability}, we can interchange the position of any two links without affecting the departure process of node $N^{(\ell)}$ and therefore the delay function. Consequently, without loss of generality we can swap the position of the worst link in the queue (that is unique from the assumptions of Theorem~\ref{thm:Expected_values_multi_hop_network}) with the first link leaving the positions of all other links unaltered, and therefore without loss of generality we can simply assume that the first link is the worst link ($p_2, p_3,\ldots, p_\ell < p_1 < 1$).

It is helpful to assume the first link to be the worst one in order to use the results of Hsu and Burke in \cite{hsu76behavior}. The authors proved that a tandem network with geometrically distributed service times and a geometric input process, reaches steady state as long as the input process is slower than any of the service times. Our line network is depicted in Fig.~\ref{fig:Multi_hop_network} and the input process (of innovative packets) is the geometric arrival process at node $N^{(2)}$ from $N^{(1)}$. Since $p_2, p_3,\ldots, p_\ell < p_1$ the arrival process is slower than any service process (transmission of the innovative packet to the next hop) and therefore the network in Fig.~\ref{fig:Multi_hop_network} reaches steady state.

Sending an arbitrarily large number of packets ($n\rightarrow \infty$) makes the problem of estimating  $\displaystyle\lim_{n\rightarrow\infty}\stexp \tau_n$--if the network was not reaching a steady state the above limit would diverge--the same as calculating the expected time taken to send all the remaining innovative packets at nodes $N^{(2)},\ldots,N^{(\ell)}$ to reach the destination $N^{(\ell+1)}$ at steady state. This is exactly the expected end-to-end delay for a single customer in a line network that has reached equilibrium. This quantity has been calculated in \cite{daduna01queueing} (page 67, Theorem 4.10) and is equal to
\begin{equation}
\lim_{n\rightarrow\infty} \stexp \tau_n =\sum_{i=2}^{\ell}\frac{p_1}{p_1-p_i}.
\label{eqn:the_expression_of_the_limit}
\end{equation}
Combining equations (\ref{eqn:inequality_of_expectations}) and (\ref{eqn:the_expression_of_the_limit}) concludes the proof of Theorem~\ref{thm:Expected_values_multi_hop_network} by changing $p_1$ to $p_m:=\max p_i < 1$.

\subsection{Proof of concentration}

Here we present a martingale concentration argument. In particular we prove a slightly stronger version of Theorem 2:

\begin{theorem}[Extended version of Theorem 2]
The time $T_n$ for $n$ packets to travel across the line network is concentrated around its expected value with high probability. In particular for sufficiently large $n$:
\begin{equation*}
\Prob[|T_n-\stexp T_n|> \epsilon_{n}] \leq \frac{2(1-\displaystyle\max_{1\leq i\leq \ell} p_i)}{n}+\frac{2(1-\displaystyle\max_{1\leq i\leq \ell} p_i))n^{2\delta}}{n^2-n^{1+2\delta}}.
\end{equation*}
for deviations $\epsilon_{n} = n^{1/2+\delta}/(1-\displaystyle\max_{1\leq i\leq \ell} p_i)$, $\delta\in (0,1/2)$.
\label{last_theorem}
\end{theorem}
\begin{proof}
The main idea of the proof is to use the method of Martingale bounded differences \cite{mitzenmacher05probability}. This method works as follows: first we show that the random variable we want to show is concentrated is a function of a finite set of independent random variables. Then we show that this function is Lipschitz with respect to these random variables, i.e. it cannot change its value too much if only one of these variables is modified. Using this function we construct the corresponding Doob martingale and use the Azuma-Hoeffding~\cite{mitzenmacher05probability} inequality to establish concentration. See also \cite{azuma1, azuma2} for related concentration results using similar martingale techniques.

Unfortunately however this method does not seem to be directly applicable to $T_n$ because it cannot be naturally expressed as a function of a \emph{bounded number} of independent random variables. We use the following trick of showing concentration for another quantity first and then linking that concentration to the concentration of $T_n$.

Specifically, we define $R_t$ to be the number of innovative (linearly independent) packets received at the destination node $N^{(\ell+1)}$ after $t$ time steps. $R_t$ is linked with $T_n$ through the equation:
\begin{equation}
T_n =\mathop{\text{arg}}_t(R_t = n)
\label{Relate_T_R}.
\end{equation}
The number of received packets is a well defined function of the link states at each time step. If there are $\ell$ links in total, then:
\begin{equation*}
R_t = g(z_{11},...,z_{1\ell},\ldots,z_{t1},...,z_{t\ell})
\label{eqn:function_of_R}
\end{equation*}
where $z_{ij}$,$1\leq i \leq t$ and $1\leq j \leq \ell$, are equal to $0$ or $1$ depending on whether link $j$ is OFF or ON at time $i$. If a packet is sent on a link that is ON, it is received successfully; if sent on a link that is OFF, it is erased. It is clear that this function satisfies a bounded Lipschitz condition with a bound equal to $1$:
\begin{eqnarray*}
|g(z_{11},...,z_{1\ell},..., z_{ij},...,z_{t1},...,z_{t\ell}) - \notag\\
g(z_{11},...,z_{1\ell},...,z_{ij}^{'} ,...,z_{t1},...,z_{t\ell})| \leq 1.
\label{eqn:g_inequality}
\end{eqnarray*}
This is because if we look at the history of all the links failing or succeeding at all the $t$ time slots, changing one of these link states in one time slot can at most influence the received rank by one.

Using the Azuma-Hoeffding inequality (see the Appendix Theorem~\ref{thm:Azuma_Hoeffding_inequality}) on the Doob martingale constructed by $R_t = g(z_{11},...,z_{1\ell},...,z_{t1},...,z_{t\ell})$ we get following the concentration result:
\begin{prop}
The number of received packets $R_t$ is a concentrated random variable around its mean value:
\begin{equation}
\Prob(|R_t-\stexp R_t|\geq \varepsilon_t)\leq \frac{1}{t}\hspace{2.5mm}\text{where}\hspace{2.5mm} \varepsilon_t\doteq \sqrt{\frac{t \ell}{2} \ell n(2t)}.
\label{eqn:concentration_of_R_2}
\end{equation}
\label{prop:concentration_of_Rt}
\end{prop}
\begin{proof}
Given in the appendix.
\end{proof}

Using this concentration and the relation (\ref{Relate_T_R}) between $T_n$ and $R_t$ we can show that deviations of the order $\varepsilon_t \doteq \sqrt{\frac{t \ell}{2} \ell n(2t)}$ for $R_t$ translate to deviations of the order of
$\epsilon_{n} = n^{1/2+\delta}/(1-\displaystyle\max_{1\leq i\leq \ell} p_i)$ for $T_n$. In Theorem~\ref{last_theorem} smaller values $\delta$ give tighter bounds that hold for larger $n$. Define the events:
\begin{equation*}
H_t=\{|R_t-\stexp R_t|< \varepsilon_t\}
\label{eqn:H_event}
\end{equation*}
and
\begin{equation*}
\overline{H}_t=\{|R_{t}-\stexp R_t| \geq \varepsilon_t\}
\label{eqn:H_hat_event_1}
\end{equation*}
and further define $t^u_n$ ($u$ stands for upper bound) to be some $t$, ideally the smallest $t$, such that $\stexp R_t -\varepsilon_t\geq n$ and $t^l_n$ ($l$ stands for lower bound) to be some $t$, ideally the largest $t$, such that $\stexp R_t+\varepsilon_t\leq n$. Then we have:
\begin{eqnarray*}
\Prob(T_n\geq t^u_n) &=& \Prob(T_n\geq t^u_n|H_{t^u_n}) \cdot \Prob(H_{t^u_n})\notag\\
                           &+& \Prob(T_n\geq t^u_n|\overline{H}_{t^u_n})\cdot \Prob(\overline{H}_{t^u_n})
\end{eqnarray*}
where:
\begin{itemize}
  \item $\Prob(T_n\geq t^u_n|H_{t^u_n}) = 0$ since at time $t = t^u_n$ the destination has already received more than $n$ innovative packets. Indeed given that $H_{t^u_n}$ holds: $n\leq \stexp R_{t^u_n} - \varepsilon_{t^u_n} < R_{t^u_n}$ where the first inequality is due to the definition of $t^u_n$.
  \item $\Prob(H_{t^u_n})\leq 1$
  \item $\Prob(T_n\geq t^u_n|\overline{H}_{t^u_n}) \leq 1$
  \item $\Prob(\overline{H}_{t^u_n})\leq \frac{1}{t^u_n}$ due to equation (\ref{eqn:concentration_of_R_2}).
\end{itemize}
Therefore:
\begin{equation}
\Prob(T_{n}\geq t^{u}_{n})\leq \frac{1}{t^{u}_{n}}.
\label{ineq:t_u_n}
\end{equation}

Similarly:
\begin{eqnarray*}
\Prob(T_n\geq t^l_n) &=& \Prob(T_n\geq t^l_n|H_{t^l_n}) \cdot \Prob(H_{t^l_n})\notag\\
                           &+& \Prob(T_n\geq t^l_n|\overline{H}_{t^l_n})\cdot \Prob(\overline{H}_{t^l_n})
\end{eqnarray*}
where:
\begin{itemize}
  \item $\Prob(T_n\leq t^l_n|H_{t^l_n}) = 0$  since at time $t = t^l_n$ the destination has already received less than $n$ innovative packets. Indeed given that $H_{t^l_n}$ holds: $R_{t^u_n} < \stexp R_{t^u_n} + \varepsilon_{t^u_n} < n$ where the last inequality is due to the definition of $t^l_n$.
  \item $\Prob(H_{t^l_n})\leq 1$
  \item $\Prob(T_n\leq t^l_n|\overline{H}_{t^l_n}) \leq 1$
  \item $\Prob(\overline{H}_{t^l_n})\leq \frac{1}{t^l_n}$ due to equation (\ref{eqn:concentration_of_R_2}).
\end{itemize}
Therefore:
\begin{equation}
\Prob(T_n\leq t^l_n)\leq \frac{1}{t^l_n}.
\label{ineq:t_l_n}
\end{equation}

Equations (\ref{ineq:t_u_n}) and (\ref{ineq:t_l_n}) show that the random variable $T_n$ representing the time required for $n$ packets to travel across a line network exhibits some kind of concentration between $t^l_n$ and $t^u_n$, which are both functions of $n$. In the case of a line network, $\stexp R_t = A \cdot t - r(t)$ where $A = (1-\displaystyle\max_{1\leq i\leq \ell} p_i)$ is a constant equal to the capacity of the line network and $r(t)$ is a bounded function representing the expected number of innovative packets that have crossed the first link (once again the worst link in the network has been positioned as the first link) by time $t$ without having reached the destination. Since $r(t)$ is bounded, a legitimate choice for large enough $n$ for $t^l_n$ and $t^u_n$ is the following (see Lemma~\ref{lemma:expressing_t_up_and_t_l} in the Appendix):
\begin{equation}
t^u_n=(n+n^{1/2+\delta'})/A,\text{ } \delta' \in (0,1/2)
\label{t_u_n_delta}
\end{equation}
\begin{equation}
t^l_n=(n-n^{1/2+\delta'})/A,\text{ } \delta' \in (0,1/2)
\label{t_l_n_delta}
\end{equation}

From both (\ref{ineq:t_u_n}) and (\ref{ineq:t_l_n}):
\begin{eqnarray}
\Prob(t^l_n\leq T_n\leq t^u_n)&=&1-\Prob(T_n \leq t^l_n)-\Prob(T_n\geq t^u_n)\notag\\
&\geq& 1-\frac{1}{t^l_n}-\frac{1}{t^u_n}
\label{ineq:Prob_t_n_1}
\end{eqnarray}
and by substituting in (\ref{ineq:Prob_t_n_1}) the $t^u_n$, $t^l_n$ from equations (\ref{t_u_n_delta}) and (\ref{t_l_n_delta}) we get:
\begin{eqnarray*}
\Prob(-\frac{n^{1/2+\delta'}}{A}\leq T_n-\frac{n}{A}\leq \frac{n^{1/2+\delta'}}{A})\geq 1 -\notag\\
\frac{A}{n-n^{1/2+\delta'}}-\frac{A}{n+n^{1/2+\delta'}}
\end{eqnarray*}
and since $\stexp T_n = \frac{n}{A}+\bigO(1)$ we have:
\begin{equation*}
\Prob(|T_n-\stexp T_n|\leq \frac{n^{1/2+\delta}}{A})\geq 1-\frac{2A}{n}-\frac{2An^{2\delta}}{n^2-n^{1+2\delta}}
\end{equation*}
or
\begin{equation*}
\Prob(|T_n-\stexp T_n| > \frac{n^{1/2+\delta}}{A})\leq \frac{2A}{n}+\frac{2An^{2\delta}}{n^2-n^{1+2\delta}}
\end{equation*}
where $\delta > \delta'$ and a simple substitution of $A$ with $(1-\displaystyle\max_{1\leq i\leq \ell} p_i)$ concludes the proof.
\end{proof}

\section{Discussion and Conclusions}
\label{Discussion}

\begin{figure}[!ht]
\begin{center}
\includegraphics[width=0.79\columnwidth]{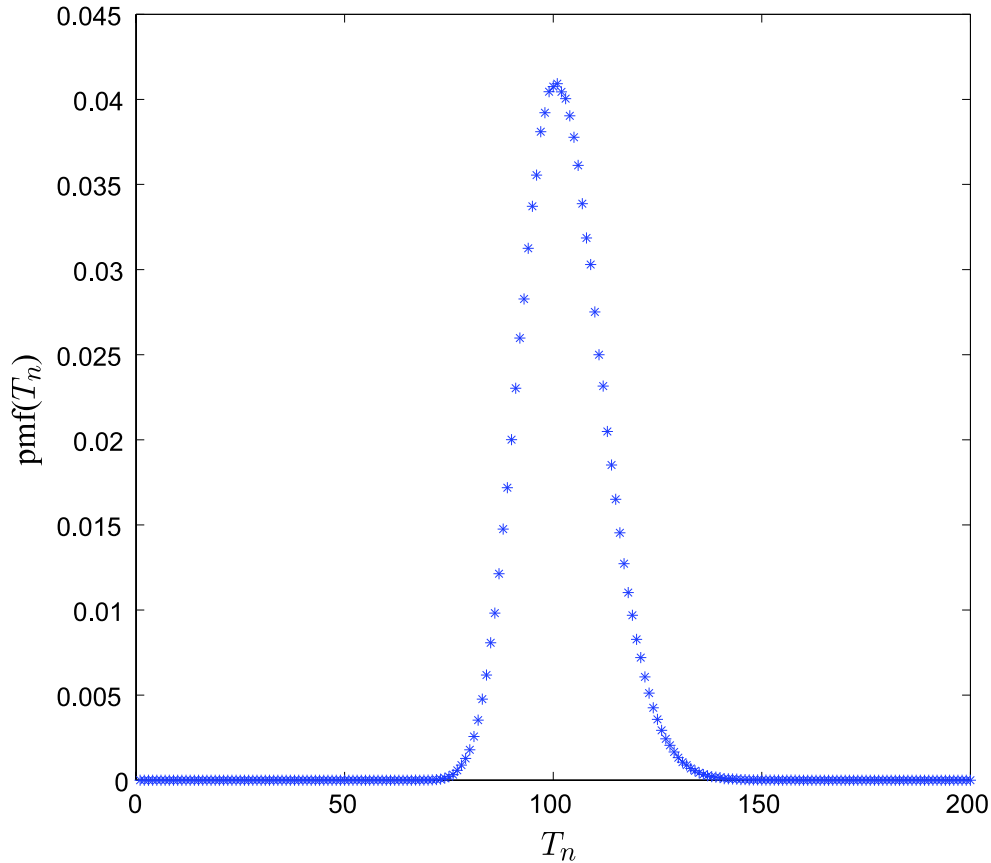}
\end{center}
\caption{The probability mass function of $T_n$ of a two-hop network with $n=50, p_1=0.5, p_2=0.3$}
\label{fig:Markov_equivalence1}
\end{figure}
\begin{figure}[!ht]
\begin{center}
\includegraphics[width=0.79\columnwidth]{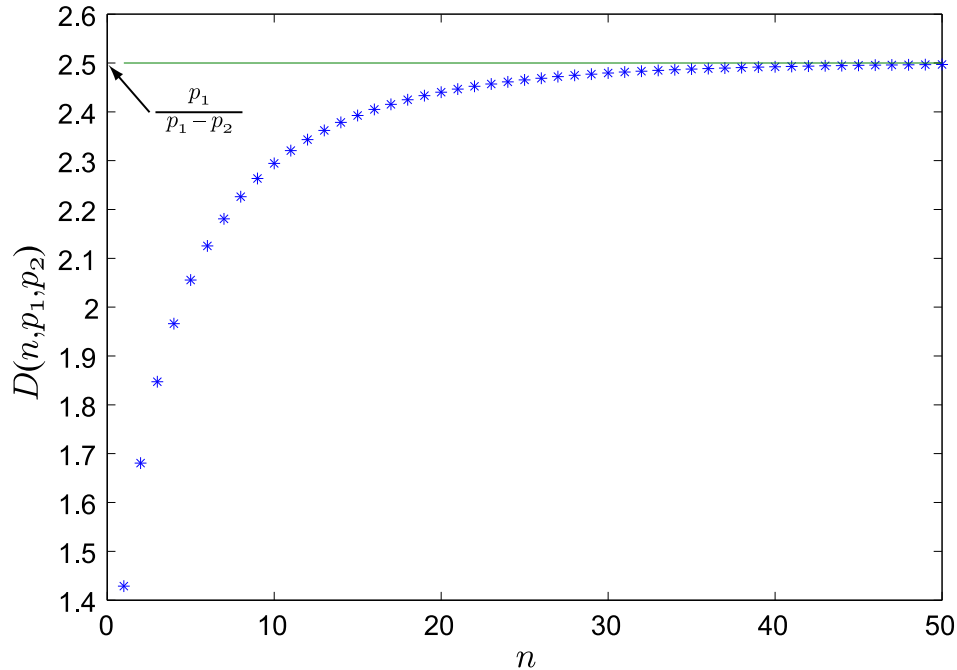}
\end{center}
\caption{The delay function $D(n,p_1,p_2)$ for a two-hop network  with $p_1=0.5, p_2=0.3$}
\label{fig:Markov_equivalence2}
\end{figure}

In this paper we analyzed the delay function and characterized its asymptotic behavior for an arbitrary set of erasure probabilities $p_1,p_2,\ldots , p_{\ell}$ that has a single worst link.
The validity of our analysis is experimentally shown in Fig. 4 and 5. In particular, Fig. 4 shows the probability mass function (pmf) --- computed via simulation --- of $T_n$ tightly concentrated around its expected value for a somewhat small value of $n=50$.
Fig. 5 shows the delay function $D(n,p_1,p_2)$ rapidly approaching the computed bound $\bar{D}(p_1,p_2)$ as $n$ grows (for $p_1=0.5$,  $p_2=0.3$).

One limitation of our technique is the assumption of a single worst
link. It is critical in  our analysis because after bringing the
worst link in the first position, it is equivalent to guaranteeing
that all the other queues are bounded in expectation. If there is
more than one bottleneck link the delay function can be
unbounded~\cite{pakzad05coding} and the general behavior remains a topic for
future work. Further understanding the delay function for more
general networks is a challenging problem that might be relevant for
delay critical  applications.

\section*{Acknowledgment}
This material is partly funded by subcontract \#069153 issued by BAE Systems National Security Solutions, Inc. and supported by the Defense Advanced Research Projects Agency
(DARPA) and the Space and Naval Warfare System Center (SPAWARSYSCEN), San Diego under Contract No. N66001-08-C-2013, and by Caltech's Lee Center for Advanced Networking.

\appendix

\begin{dfn}
A binary relation $\preceq$ defined on a set $P$ is called a preorder if it is reflexive and transitive, i.e. $\forall a, b, c \in P$:
\begin{eqnarray}
&a\preceq a &\text{ (reflexivity)}\\
&(a\preceq b)\wedge(b\preceq c) \Rightarrow a\preceq c &\text{ (transitivity)} 
\end{eqnarray}
\end{dfn}

\begin{dfn}
On the set $\mathbb{N}\hspace{0.8mm}^{\ell-1}$ of all integer $(\ell-1)$-tuples we define the regular preorder $\preceq$ that is $\forall a,b\in\mathbb{N}\hspace{0.8mm}^{\ell-1}$ $a\preceq b$ iff $a_1\leq b_1,\ldots,a_{\ell-1}\leq b_{\ell-1}$ where $a=(a_1,\ldots,a_{\ell-1})$ and $b=(b_1,\ldots,b_{\ell-1})$. Similarly we can define the preorder $\succeq$.  
\end{dfn}

\begin{dfn}
\label{dfn:usual_stochastic_order}
A random vector $X\in \mathbb{N}\hspace{0.8mm}^{\ell-1}$ is said to be stochastically smaller in the usual stochastic order than a random vector  $Y \in \mathbb{N}\hspace{0.8mm}^{\ell-1}$, (denoted by $X \preceq_{\text{st}} Y$) if: $\forall \omega \in \mathbb{N}\hspace{0.8mm}^{\ell-1}$, $\Prob(X\succeq \omega)\leq \Prob(Y\succeq \omega)$.
\end{dfn}

\begin{dfn}
\label{dfn:increasing_function}
A family of random variables $\{Y_n\}_{n\in\mathbb{N}}$ is called stochastically increasing ($\preceq_{\text{st}}$-increasing) if $Y_k \preceq_{\text{st}} Y_n$ whenever $k\leq n$.
\end{dfn}

\begin{proof}[Proof of Proposition~\ref{prop:The_markov_chain}]
Markov process $\{Y_n, n\geq 1\}$,  is a multidimensional process on $E=\mathbb{N}\hspace{0.8mm}^{\ell-1}$ representing the number of innovative packets at nodes $N^{(2)},\ldots,N^{(\ell)}$ when packet $n$ arrives at $N^{(2)}$. To prove that the Markov process $\{Y_n, n\geq 1\}$ is stochastically increasing we introduce two other processes $\{X_n, n\geq 1\}$ and $\{Z_n, n\geq 1\}$ having the same state space and transition probabilities as $\{Y_n, n\geq 1\}$.

More precisely, Markov process $\{Y_n, n\geq 1\}$ is effectively observing the evolution of the number of innovative packets present at every node of the tandem queue. We define the two new processes $\{X_n, n\geq 1\}$ and $\{Z_n, n\geq 1\}$ to observe the evolution of two other tandem queues having the same link failure probabilities as the queue of $\{Y_n, n\geq 1\}$.

\begin{figure}[!ht]
\begin{center}
\includegraphics[clip=true, trim=0mm 0mm 0mm 0mm, width=\columnwidth]{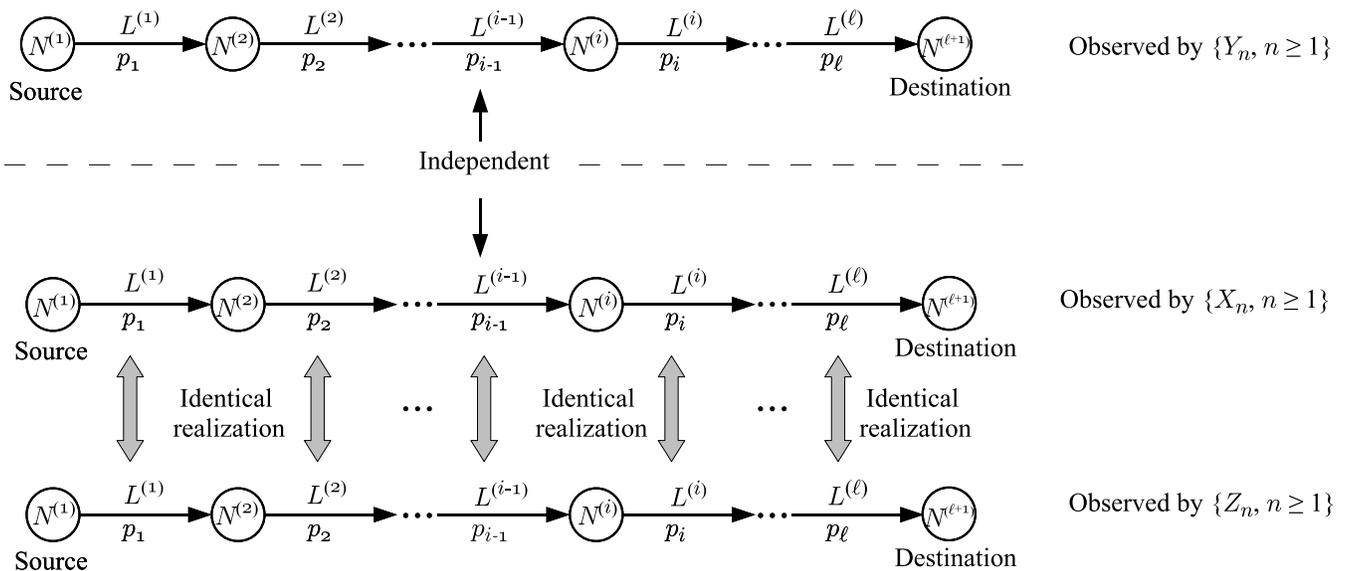}
\end{center}
\caption{Multi-hop network with the corresponding Markov chains}
\label{fig:Multi_hop_network_observed}
\end{figure} 

As seen in Fig.~\ref{fig:Multi_hop_network_observed}, at each time step and at every link, the queues for $\{X_n, n\geq 1\}$ and $\{Z_n, n\geq 1\}$ either both succeed or a fail together. Moreover the successes or failures on each link on the queues observed by $\{X_n, n\geq 1\}$ and $\{Z_n, n\geq 1\}$ are independent of the successes or failures on the queue observed by $\{Y_n, n\geq 1\}$. Formally the joint process $\{(X_n, Z_n), n\geq 1\}$ constitute a coupling meaning that marginally each one of $\{X_n, n\geq 1\}$ and $\{Z_n, n\geq 1\}$ have the transition matrix $\Prob_Y$ of $\{Y_n, n\geq 1\}$. If Markov processes $\{X_n, n\geq 1\}$ and $\{Z_n, n\geq 1\}$ have different initial conditions then the following relation holds: 
\begin{eqnarray}
\label{eqn:corollary_for_markov_chain}
X_1\preceq Z_1 \Rightarrow X_n\preceq Z_n
\end{eqnarray}

The proof of the above statement is very similar to the proof of Proposition 2 in \cite{aggregate}. Essentially relation (\ref{eqn:corollary_for_markov_chain}) states that since at both queues all links succeed or fail together the queue that holds more packets at each node initially ($n=1$) will also hold more packets subsequently ($n > 1$) at every node.

The initial state $Y_1$ of Markov process $\{Y_n, n\geq 1\}$ is state $\alpha = (1,0,\ldots,0)$ that is also called the minimal state since any other state is greater than the minimal state. To prove Proposition~\ref{prop:The_markov_chain} we set both processes $\{Y_n, n\geq 1\}$ and $\{X_n, n\geq 1\}$ to start from the minimal state ($Y_1\displaystyle \mathop{=}^\mathcal{D}\delta_\alpha, X_1\displaystyle \mathop{=}^\mathcal{D}\delta_\alpha\text{ where }\displaystyle\mathop{=}^\mathcal{D}$ means equality in distribution), whereas process $\{Z_n, n\geq 1\}$ has initial distribution $\mu$ that is the distribution of process $\{Y_n, n\geq 1\}$ after $(n-k)$ steps $(\mu=\Prob_Y^{n-k}\delta_\alpha\text{ and } Z_1\displaystyle \mathop{=}^\mathcal{D}\mu$). Then for every $\omega$ in the state space of $\{Y_n, n\geq 1\}$ we get: 
\begin{eqnarray}
\label{eqn:lemma_condition_for_stochastic_ordering_1}
\Prob(X_n\succeq \omega)=\Prob(Y_n\succeq \omega)=\Prob(Z_k\succeq \omega)
\end{eqnarray}
where the first equality holds since the two processes have the same distribution--both start from the minimal element and have the same transition matrices--and the second equality holds since
\begin{eqnarray*}
\label{eqn:lemma_condition_for_stochastic_ordering_2}
\displaystyle Z_k\mathop{=}^\mathcal{D} \Prob_Y^k \mu\equiv \Prob_Y^k (\Prob_Y^{n-k}\delta_\alpha)=\Prob_Y^n\delta_\alpha\mathop{=}^\mathcal{D}Y_n.
\end{eqnarray*}
Moreover due to the definition of the minimal element, $X_1\preceq Z_1$ and using (\ref{eqn:corollary_for_markov_chain}) we get $X_n \preceq Z_n$. Therefore
\begin{eqnarray}
\label{eqn:lemma_condition_for_stochastic_ordering_3}
\Prob(Z_k\succeq \omega)\geq \Prob(X_k \succeq \omega)=\Prob(Y_k\succeq \omega).
\end{eqnarray}
The last equality follows from the fact that the two distributions have the same law. Equations (\ref{eqn:lemma_condition_for_stochastic_ordering_1}) and  (\ref{eqn:lemma_condition_for_stochastic_ordering_3}) conclude the proof.
\end{proof}

\begin{dfn}
A sequence of random variables $V_0,V_1,\ldots$ is said to be a \textbf{martingale with respect to} another sequence $U_0,U_1,\ldots$ if, for all $n\geq 0$, the following conditions hold:
\begin{itemize}
  \item $\stexp[|V_n|]<\infty $
  \item $\stexp[V_{n+1}|U_0,\ldots,U_n]=V_n$
\end{itemize}
A sequence of random variables $V_0,V_1,\dots$ is called \textbf{martingale} when it is a martingale with respect to itself. That is:
\begin{itemize}
  \item $\stexp[|V_n|]<\infty$
  \item $\stexp[V_{n+1}|V_0,...,V_n]=V_n$
\end{itemize}
\end{dfn}

\begin{theorem}
(Azuma-Hoeffding Inequality): Let $X_{0}$, $X_{1}$,...,$X_{n}$ be a martingale such that
\[B_{k}\leq X_{k}-X_{k-1} \leq B_{k}+d_{k}\]
for some constants $d_{k}$ and for some random variables $B_{k}$ that may be a function of $X_{0},...,X_{k-1}$. Then for all $t\geq 0$ and any $\lambda > 0$,
\[\Prob(|X_{t}-X_{0}|\geq \lambda)\leq 2\exp\left(-\frac{2\lambda^2}{\sum_{i=1}^{t}d_{i}^{2}}\right)\]
\label{thm:Azuma_Hoeffding_inequality}
\end{theorem}
\begin{proof}
Theorem 12.6 in \cite{mitzenmacher05probability}
\end{proof}

\begin{proof}[Proof of Proposition~\ref{prop:concentration_of_Rt}]
The proof is based on the fact that from a sequence of random variables $U_1,U_2,\ldots,U_n$ and any function $f$ it's possible to define a new sequence $V_0,\ldots,V_n$
\[\left\{
  \begin{array}{ll}
    V_0=\stexp[f(U_1,\ldots,U_n)]\\
    V_i= \stexp[f(U_1,\ldots,U_n)|U_1,\ldots,U_i]
  \end{array}
\right.\]
that is a martingale (\textit{Doob} martingale). Using the identity $\stexp[V|W]=\stexp[\stexp[V|U,W]| W]$ it's easy to verify that the above sequence $V_0,\ldots,V_n$ is indeed a martingale. Moreover if function $f$ is \textit{c-Lipschitz} and $U_1,\ldots,U_n$ are independent it can be proved that the differences $V_i-V_{i-1}$ are restricted within bounded intervals \cite{mitzenmacher05probability} (pages 305-306). 

Function $R_t=g(z_{11},...,z_{t\ell})$ has a bounded expectation, is \textit{1-Lipschitz} and the random variables $z_{ij}$ are independent and therefore all the requirements of the above analysis hold. Specifically by setting
\begin{eqnarray*}
G_{h}=\stexp[g(z_{11},...,z_{t\ell})&|\underbrace{z_{11},...,z_{kr}}]\\                                   
                                    &\text{$h$-terms in total}
\end{eqnarray*}
we can apply the Azuma-Hoeffding inequality on the $G_{0},...,G_{t\ell}$ martingale and we get the following concentration result
\begin{eqnarray}
\Prob[|G_{t\ell}-G_{0}|\geq \lambda]= \Prob[|R_{t}-\stexp[R_{t}]|\geq \lambda] \leq 2\exp \{ -\frac{2\lambda^2}{t\ell} \}.
\label{eqn:intermediate_concentration_result}
\end{eqnarray}
The equality above holds since
\begin{itemize}
  \item $G_{0}    =\stexp[R_{t}]$
  \item $G_{t\ell}=R_{t}\text{ (the random variable itself)}$
\end{itemize}
and by substituting on (\ref{eqn:intermediate_concentration_result})  $\lambda$ with $\varepsilon_t\doteq\sqrt{\frac{t\ell}{2}\ell n(2t)}$
\begin{equation*}
\Prob[|R_{t}-\stexp[R_{t}]|\geq \varepsilon_t] \leq \frac{1}{t}
\end{equation*}
\end{proof}

\begin{lemma}
When the expected number of innovative packets $\stexp R_t$ received at the destination  by time $t$ is given by $\stexp R_t = A\cdot t- r(t)$ where $A$ is a constant and $r(t)$ is a bounded function then one legitimate choice for $t^u_n$ and $t^l_n$ is:
\begin{eqnarray*}
t^u_n=(n+n^{1/2+\delta'})/A,\text{ } \delta' \in (0,1/2)\\
t^l_n=(n-n^{1/2+\delta'})/A,\text{ } \delta' \in (0,1/2)
\end{eqnarray*} 
\label{lemma:expressing_t_up_and_t_l}
\end{lemma}
\begin{proof}
The only requirement for $t^u_n$ is that it is a $t$ such that $\stexp R_t-\epsilon_t\geq n$. This is indeed true for large enough $n$ if we substitute $t^u_n$ with $(n+n^{1/2+\delta'})/A$:
\begin{eqnarray*}
\stexp[R_{t^u_n}]-\epsilon_{{t^u_n}} \geq n \Rightarrow At^u_n-r(t^u_n)-\epsilon_{t^u_n}\geq n\Rightarrow At^u_n-r(t^u_n)-\sqrt{\frac{\ell\cdot t^u_n}{2} \ell n(2t^u_n)}\geq n\\
\Rightarrow A\cdot \frac{n+n^{1/2+\delta}}{A}-r(t^u_n)- \sqrt{\frac{\ell(n+n^{1/2+\delta)}}{2A} \ell n(\frac{2(n+n^{1/2+\delta})}{A})}\geq \\
\geq n+n^{1/2+\delta}-B- \sqrt{\frac{\ell(n+n^{1/2+\delta)}}{2A} \ell n(\frac{2(n+n^{1/2+\delta})}{A})} \geq n \\
\Rightarrow n^{1/2+\delta}\geq \sqrt{\frac{\ell(n+n^{1/2+\delta)}}{2A} \ell n(\frac{2(n+n^{1/2+\delta})}{A})}+B \Rightarrow n^{1/2+\delta}\geq \sqrt{n} \sqrt{\frac{\ell(1+n^{\delta-1/2)}}{2A} \ell n(\frac{2(n+n^{1/2+\delta})}{A})}+B\\
\Rightarrow  n^{\delta} \geq \sqrt{\frac{\ell(1+n^{\delta-1/2)}}{2A} \ell n(\frac{2(n+n^{1/2+\delta})}{A})}+\frac{B}{n^{1/2}}
\end{eqnarray*}
where $B$ is the upper bound of the function $r(t)$ and the last equation holds for large enough $n$.

Similarly $t^l_n$ is a $t$ such that $\stexp R_t+\epsilon_t\leq n$. This is indeed true for large enough $n$ if we substitute $t^l_n$ with $(n-n^{1/2+\delta'})/A$:
\begin{eqnarray*}
\stexp[R_{t^l_n}]+\epsilon_{{t^l_n}} \leq n \Rightarrow At^l_n-r(t^l_n)+\epsilon_{t^l_n}\leq n\Rightarrow At^l_n-r(t^l_n)+\sqrt{\frac{\ell\cdot t^l_n}{2} \ell n(2t^l_n)}\leq n\\
\Rightarrow A\cdot \frac{n-n^{1/2+\delta}}{A}-r(t^l_n)+ \sqrt{\frac{\ell(n-n^{1/2+\delta)}}{2A} \ell n(\frac{2(n-n^{1/2+\delta})}{A})}\leq \\
\leq n-n^{1/2+\delta}+ \sqrt{\frac{\ell(n-n^{1/2+\delta)}}{2A} \ell n(\frac{2(n-n^{1/2+\delta})}{A})} \leq n \\
\Rightarrow \sqrt{\frac{\ell(n-n^{1/2+\delta)}}{2A} \ell n(\frac{2(n-n^{1/2+\delta})}{A})} \leq n^{1/2+\delta}\Rightarrow  \sqrt{n} \sqrt{\frac{\ell(1-n^{\delta-1/2)}}{2A} \ell n(\frac{2(n-n^{1/2+\delta})}{A})} \leq n^{1/2+\delta}\\
\Rightarrow \sqrt{\frac{\ell(1-n^{\delta-1/2)}}{2A} \ell n(\frac{2(n-n^{1/2+\delta})}{A})} \leq n^{\delta}
\end{eqnarray*}
where the last inequality holds for large enough $n$.
\end{proof}

\bibliographystyle{IEEEtran}
\bibliography{IEEEabrv,NWC-abbr}

\begin{thebibliography}{10}
\providecommand{\url}[1]{#1}
\csname url@samestyle\endcsname
\providecommand{\newblock}{\relax}
\providecommand{\bibinfo}[2]{#2}
\providecommand{\BIBentrySTDinterwordspacing}{\spaceskip=0pt\relax}
\providecommand{\BIBentryALTinterwordstretchfactor}{4}
\providecommand{\BIBentryALTinterwordspacing}{\spaceskip=\fontdimen2\font plus
\BIBentryALTinterwordstretchfactor\fontdimen3\font minus
  \fontdimen4\font\relax}
\providecommand{\BIBforeignlanguage}[2]{{%
\expandafter\ifx\csname l@#1\endcsname\relax
\typeout{** WARNING: IEEEtran.bst: No hyphenation pattern has been}%
\typeout{** loaded for the language `#1'. Using the pattern for}%
\typeout{** the default language instead.}%
\else
\language=\csname l@#1\endcsname
\fi
#2}}
\providecommand{\BIBdecl}{\relax}
\BIBdecl

\bibitem{chou03practical}
\BIBentryALTinterwordspacing
P.~A. Chou, Y.~Wu, and K.~Jain, ``Practical network coding,'' in \emph{Proc.
  41st Annual Allerton Conference on Communication, Control, and Computing},
  2003. [Online]. Available:
  \url{http://research.microsoft.com/~pachou/pubs/ChouWJ03.pdf}
\BIBentrySTDinterwordspacing

\bibitem{lun04coding}
D.~S. Lun, M.~M\'edard, and M.~Effros, ``On coding for reliable communication
  over packet networks,'' in \emph{In Proc. 42nd Annual Allerton Conference on
  Communication, Control, and Computing, Invited paper}, September-October
  2004.

\bibitem{pakzad05coding}
P.~Pakzad, C.~Fragouli, and A.~Shokrollahi, ``Coding schemes for line
  networks,'' in \emph{Proc. {IEEE} Int. Symp. Inf. Theory ({ISIT})}, Sep.
  2005, pp. 1853--1857.

\bibitem{dana06capacity}
A.~F. Dana, R.~Gowaikar, R.~Palanki, B.~Hassibi, and M.~Effros, ``Capacity of
  wireless erasure networks,'' \emph{IEEE Transactions on Information Theory},
  vol.~52, pp. 789--804, 2006.

\bibitem{rubin74communication}
I.~Rubin, ``Communication networks: Message path delays,'' \emph{{IEEE} Trans.
  Inf. Theory}, vol.~20, no.~6, pp. 738--745, Nov. 1974.

\bibitem{shalmon87exact}
M.~Shalmon, ``Exact delay analysis of packet-switching concentrating
  networks,'' \emph{{IEEE} Trans. Commun.}, vol.~35, no.~12, pp. 1265--1271,
  Dec. 1987.

\bibitem{ephremides2006}
B.~Shrader and A.~Ephremides, ``On the queueing delay of a multicast erasure
  channel,'' in \emph{Proceedings of the IEEE Information Theory Workshop},
  2006.

\bibitem{jay2007}
J.~Sundararajan, D.~Shah, and M.~M\'edard, ``On queueing in coded networks
  queue size follows degrees of freedom,'' \emph{Information Theory for
  Wireless Networks, 2007 IEEE Information Theory Workshop on}, pp. 1--6, July
  2007.

\bibitem{stock2}
T.~Lindvall, \emph{Lectures on the Coupling Method}.\hskip 1em plus 0.5em minus
  0.4em\relax Courier Dover Publications, 2002.

\bibitem{aggregate}
H.~Castel-Taleb, L.~Mokdad, and N.~Pekergin, ``Aggregated bounding markov
  processes applied to the analysis of tandem queues,'' in \emph{ValueTools
  '07: Proceedings of the 2nd international conference on Performance
  evaluation methodologies and tools}.\hskip 1em plus 0.5em minus 0.4em\relax
  ICST, 2007, pp. 1--10.

\bibitem{weber92interchangeability}
R.~R. Weber, ``The interchangeability of tandem queues with heterogeneous
  customers and dependent service times,'' \emph{Adv. Appl. Probability},
  vol.~24, no.~3, pp. 727--737, Sep. 1992.

\bibitem{hsu76behavior}
J.~Hsu and P.~Burke, ``Behavior of tandem buffers with geometric input and
  {Markovian} output,'' \emph{{IEEE} Trans. Commun.}, vol.~24, no.~3, pp.
  358--361, Mar. 1976.

\bibitem{daduna01queueing}
H.~Daduna, \emph{Queueing Networks with Discrete Time Scale}.\hskip 1em plus
  0.5em minus 0.4em\relax New York: Springer-Verlag, 2001.

\bibitem{mitzenmacher05probability}
M.~Mitzenmacher and E.~Upfal, \emph{Probability and Computing: Randomized
  Algorithms and Probabilistic Analysis}.\hskip 1em plus 0.5em minus
  0.4em\relax Cambridge University Press, 2005.

\bibitem{azuma1}
\BIBentryALTinterwordspacing
Z.~Kong, S.~A. Aly, E.~Soljanin, E.~M. Yeh, and A.~Klappenecker, ``Network
  coding capacity of random wireless networks under a sinr model,'' 2008.
  [Online]. Available:
  \url{http://www.citebase.org/abstract?id=oai:arXiv.org:0804.4284}
\BIBentrySTDinterwordspacing

\bibitem{azuma2}
S.~A. Aly, V.~Kapoor, and J.~Meng, ``Bounds on the network coding capacity for
  wireless random networks,'' in \emph{In Proc. 3rd Workshop on Network Coding,
  Theory, and Applications}, 2007.

\end{thebibliography}
\end{document}